\documentclass[]{article}
\usepackage{bbm}
\usepackage{bm}
\usepackage[T1]{fontenc}
\usepackage[utf8]{inputenc}
\usepackage{lmodern}
\usepackage{fullpage}
\usepackage{url}
\usepackage{amsmath,amssymb,amsthm}
\newtheorem{theorem}{Theorem}
\newtheorem{corollary}[theorem]{Corollary}
\newtheorem{lemma}[theorem]{Lemma}

\newtheorem{definition}{Definition}

\usepackage{xspace}

\usepackage{cleveref}

\usepackage[linesnumberedhidden,vlined]{algorithm2e}

\newcommand{\Prob}[1]{\mathbf{Pr}\left[#1\right]}

\def\epsilon{\ensuremath{\varepsilon }}
\newcommand{\eps}{\ensuremath{\epsilon }}

\newcommand{\Exp}[1]{\mathbf{E}\left[#1\right]}

\newcommand{\COMMENTED}[1]{{}}
\newcommand{\hide}[1]{\COMMENTED{#1}}

\newcommand{\ratio}{\ensuremath{\xi}}
\newcommand{\ratiod}{\ensuremath{\ratio^{2d}}}

\newcommand{\local}{\textsf{LOCAL}\xspace}

\newcommand{\congest}{\textsf{CONGEST}\xspace}
\newcommand{\bcongest}{\textsf{Broadcast CONGEST}\xspace}

\newcommand{\const}{\ell}
\date{}

\author{Peter Davies-Peck\\Durham University\\\url{peter.w.davies@durham.ac.uk}}

\title{Distributed Approximate Maximum Matching and Minimum Vertex Cover via Generalized Graph Decomposition}

\begin{document}

\maketitle

\begin{abstract}
The classic lower bound of Kuhn, Moscibroda and Wattenhofer [JACM 2016] states that approximate maximum matching and approximate vertex cover (among other problems) in the LOCAL model require $\Omega(\min\{\sqrt{\frac{\log n}{\log\log n}}, \frac{\log \Delta}{\log\log \Delta}\})$ rounds, for any polylogarithmic or smaller approximation ratio. As a function of $\Delta$, this complexity was subsequently matched for constant-approximate weighted vertex cover [Bar-Yehuda, Censor-Hillel and Schwartzman, JACM 2017] and constant-approximate maximum matching [Bar-Yehuda, Censor-Hillel, Ghaffari and Schwartzman, PODC 2017]. One might expect, therefore, that the true complexity should be $\Theta(\frac{\log \Delta}{\log\log \Delta})$, and the $n$-dependent term in the lower bound is just an artefact of the proof method.

We show that this is not the case, and a term dependent on $n$ is in fact required. Specifically, we show randomized algorithms for $2+\eps$-approximate maximum matching and approximate (weighted) minimum vertex cover taking $O(\frac{\log n}{\log^2 \log n})$ rounds. Our algorithms are based on a novel graph decomposition result generalizing the method of Miller, Peng and Xu [SPAA 2013], which we use to reduce the `effective' degree of high-degree graphs. We expect that this decomposition may be of further use for other problems.

\end{abstract}
\section{Introduction}

We study the central problems of approximate maximum matching and approximate vertex cover in the classic \local model of distributed computing. The \local\ model is a graph-based, message passing model, which proceeds in synchronous rounds. The $n$ nodes of the graph are the processors, and to begin with they do not have any knowledge about the input graph (other than, generally, the value of $n$, or at least a polynomial upper bound). If the problem involves input (for example node or edge weights), then nodes only know their own local part of the input (i.e. their own weight and/or the weights of their incident edges). In each round, nodes can perform any arbitrary computation, and send messages of arbitrary size, but only to their neighbors in the input graph - that is, the communication graph is the same as the input graph. The goal is to allow, in as few rounds as possible, all nodes to output their own local part of the required output, so that overall the output is consistent and correct.

The problem of approximate maximum matching is to output a matching (a set of non-adjacent edges) of size within a certain approximation factor of the maximum possible. The problem of approximate minimum vertex cover is to output a set of nodes that are adjacent to all edges, where the set is within the approximation factor of the minimum possible size. In the weighted version, nodes are given non-negative weights as input, and it is the sum of weights of included nodes that should be minimized rather than the number of included nodes.

Maximum matching and minimum vertex cover are highly related problems, with a linear programming formulation in which they are dual programs. The exact problems are NP-complete (minimum vertex cover was one of Karp's 21 NP-complete problems \cite{K72}), and can easily be shown to require $\Omega(n)$ rounds in \local. Distributed algorithms therefore focus on approximations. Of particular interest is the problem of maximal matching, often considered one of the most central problems in the model. A maximal matching is a $2$-approximate maximum matching, and the vertex set given by taking all endpoints of matched edges is a $2$-approximate minimum vertex cover.

\subsection{Related Work}

A classic paper of Israeli and Itai \cite{II86} gave an $O(\log n)$-round algorithm for maximal matching (and therefore also $2$-approximate maximum matching and $2$-approximate minimum vertex cover). This paper pre-dated the formalization of the \local\ model, but it can be easily seen that the algorithm can be implemented in \local. The dedicated study of distributed algorithms for approximate matching and vertex cover began with Lotker, Patt-Shamir and Pettie \cite{LPP15}, who gave an $O(\eps^{-3}\log n)$-round randomized algorithm for $1+\eps$-approximate maximum matching, as well as an $O(\log \eps^{-1}\log n)$-round algorithm for $2+\eps$-approximate maximum \emph{weighted} matching. These results were both improved by Bar-Yehuda, Censor-Hillel, Ghaffari and Schwartzmann \cite{BCGS17} to $O(\frac{\log\Delta}{\log\log\Delta})$ rounds (for constant $\eps$), still randomized. Maximal matching has also seen a long history of study in \local; currently, the fastest known randomized algorithm takes $O(\log \Delta + \log^{O(1)}\log n)$ \cite{BEPS16}, and the fastest deterministic algorithm takes $\tilde O(\log^{5/3} n)$ rounds \cite{GG24}.

For minimum weighted vertex cover, Bar-Yehuda, Censor-Hillel and Schwartzman \cite{BCS17} gave a deterministic \local\ $2+\eps$-approximation algorithm running in $O(\frac{\log\Delta}{\eps\log\log\Delta})$ rounds. Subsequently, Ben-Basat, Even, Kawarabayashi and Schwartzman \cite{BEKS18} improved the dependency on $\eps$ to $O(\frac{\log\Delta \log 1/\eps}{\log\log\Delta})$. There is also a line of work in the low-degree regime, such as the $2$-approximate minimum weighted vertex cover of \r{A}strand and Suomela \cite{AS10} running in $O(\Delta + \log^* W)$ rounds (when weights are from $\{1,2,\dots,W\}$).

On the lower bound side, the classic lower bound of Kuhn, Moscibroda, and Wattenhofer \cite{KMW16} implies that, for the problems of approximate maximum matching and approximate minimum vertex cover, $\Omega(\sqrt{\frac{\log n}{\log\log n}})$ rounds are needed for any $\log^{O(1)} n$ approximation ratio, and $\Omega(\frac{\log \Delta}{\log\log \Delta})$ rounds are needed for any $\log^{O(1)} \Delta$ approximation ratio. To achieve, for example, a constant approximation ratio, therefore, $\Omega(\min\{\sqrt{\frac{\log n}{\log\log n}}, \frac{\log \Delta}{\log\log \Delta}\})$ rounds are required. For maximal matching, Khoury and Schild \cite{KS25} recently improved the lower bound to $\Omega(\min\{\log \Delta, \sqrt{\log n}\})$, but this bound does not apply to approximate maximum matching or approximate minimum vertex cover. 

The reason for the $\sqrt{\frac{\log n}{\log\log n}}$ term in the lower bound of \cite{KMW16} is that it is based on an indistinguishability argument on nodes' neighborhoods, and in order to reason about these neighborhoods, the bound is proven on a high-girth graph so that the neighborhoods are trees. This high-girth requirement effectively imposes an upper bound on the maximum degree for which the $\Omega(\frac{\log \Delta}{\log\log \Delta})$ lower bound works - for example, a $\Delta$-regular graph of girth $g$ contains $\Delta^{\Omega(g)}$ nodes, and therefore must have $\Delta= n^{O(1/g)}$ (the graph used in \cite{KMW16} is not regular, but the same intuition applies). It was unclear whether this high-girth property was really necessary, or whether it might be possible to extend the bound to $\Omega(\frac{\log \Delta}{\log\log \Delta})$ for all $\Delta$ by using low-girth graphs and an (even) more sophisticated argument about the isomorphism between nodes' neighborhoods. Our work shows that there must be a limit to such an approach, and that in high-degree graphs one can leverage structural results to surpass the $\Theta(\frac{\log \Delta}{\log\log \Delta})$ barrier.

\subsection{Our Results}

Our main results are randomized \local algorithms that achieve $2+\eps$-approximations of  maximum (unweighted) matching and minimum weighted vertex cover in $O(\frac{\log n}{\log^2\log n})$ rounds: 

\begin{theorem}\label{thm:MM}
	For any $\eps \in (0,1)$ (which can be a function of $n$) and any constant $\delta>0$, a $2+\eps $-approximate maximum cardinality matching can be computed in $O(\frac{\log n}{\log^2\log n})$ rounds of \local, succeeding with probability at least $1-\frac{1}{\eps} e^{- \Omega(\log^{1-\delta} n)}$.
\end{theorem}

\begin{theorem}\label{thm:MWVC}
	For any	$\eps \in (0,1)$ (which can be a function of $n$), a $2+\eps$-approximate minimum weight vertex cover can be computed in $O(\frac{\log n}{\log^2\log n} + \frac{\log n\log 1/\eps}{\log^3\log n})$ rounds of \local, succeeding with high probability in $n$.
\end{theorem}

While the round complexity may seem a modest improvement over the previous best results (a $\log\log n$-factor improvement over \cite{BCGS17,BCS17,BEKS18}, and only for large $\Delta$), these results are significant since they demonstrate that an $n$-dependent term in the complexity is necessary, and the $O(\frac{\log \Delta}{\log\log\Delta})$ bound is not tight for all $\Delta$.

These algorithms are based on a novel method of graph decomposition that generalizes the well-known procedure of Miller, Peng and Xu \cite{MPX13}. This is a randomized procedure that uses exponentially-distributed shift variables to partition the graph into clusters, each with their own cluster center, that have desirable properties such as low diameter/radius\footnote{Radius here means the maximum distance from a cluster's center to any of its members, meaning that the diameter is at most twice the radius.}.

Here, we show a new analysis of the procedure that allows the use of any distribution for the shift variables, rather than requiring the memorylessness property of the exponential distribution. We then use a particular distribution (with a \emph{polynomially}, rather than exponentially, decreasing tail) to yield a decomposition with the following properties:

\newcommand{\diff}{\ensuremath{\bm{d}}}
\newcommand{\z}{\ensuremath{\mathcal Z}}
\newcommand{\y}{\ensuremath{\mathcal Y}}
\newcommand{\mult}{\ensuremath{\pi}}
\newcommand{\add}{\ensuremath{\tau}}

\begin{lemma}\label{lem:localpart}
	For any constant $\const \ge 4$, and for $\alpha = \frac{\const\log n}{\log\log n}$, there is a randomized \local (and \congest) algorithm, running in $\log^{3/\const} n$ rounds, that produces a clustering with the following properties, with high probability:
	
	\begin{itemize}
		\item All clusters $C$ have radius $r_C \le \log^{3/\const} n$, and so diameter at most $2\log^{3/\const} n$. 
		\item For any $q\ge 0$, each node is adjacent to at most $e^{\frac{3\alpha}{q+1}}$ clusters $C$ with radius $r_C\ge q$.
	\end{itemize}
\end{lemma}

Most prior applications of the clustering of \cite{MPX13} use logarithmic or higher diameter; here we focus on a clustering of sublogarithmic diameter and show that this can still provide useful properties for distributed graph algorithms. The intuition behind the use of this clustering is that some algorithms can be adapted to have some decisions made by clusters rather than nodes. Since clusters need to aggregate information to their centers in order to perform the computation necessary to make these decisions, the algorithm is slowed down by a factor of $O(q)$, where $q$ is the radius of the cluster. However, the number of clusters of radius $q$ adjacent to any node (which one can think of as the effective degree) is only $2^{O(\frac{\alpha}{q})}$. So, if an algorithm with a round complexity that is a function of maximum degree $T(\Delta)$ can be adapted in this way, it would require $O(qT(2^{O(\frac{\alpha}{q})}))$ rounds. We adapt algorithms with $T(\Delta) = O(\frac{\log\Delta}{\log\log\Delta})$ complexities, and so we end up with complexities of

\[O(\frac{q\log 2^{O(\frac{\alpha}{q})}}{\log\log2^{O(\frac{\alpha}{q})}}) = O(\frac{\alpha}{\log\frac{\alpha}{q}}) = O(\frac{\log n}{\log^2\log n})\enspace.\]

However, the adaptations required to the algorithms are not straightforward - significant changes are required to the previous algorithms for both problems in order to preserve critical functionality of the algorithms while passing some of the decisions to clusters. One major issue is that the process described above happens in parallel for all values of $q$, from $0$ to $\log^{3/\const} n$, which adds a level of `asynchrony' that must be accounted for in the algorithm.

\section{Distributed Approximate Maximum Matching}
In this section we give our approximate maximum matching algorithm. This algorithm is derived by adapting the following simple framework (Algorithm \ref{alg:basicmatch}): a fractional matching is maintained by setting all edge weights initially to $\frac{1}{2\Delta}$, and iteratively raising them by a multiplicative factor of $K$ only if the sum of weights around each endpoint is at most $\frac{1}{2K}$. This ensures that the sum of weights around any vertex will always be at most $\frac12$. Also, in every two rounds, edges attempt to join the matching with probability equal to their weight, and are successful if none of their neighbors also attempt to join in the same round (this takes two rounds per iteration because it takes a communication round for nodes to determine if they successfully joined, and a second round to inform their neighbors).

\RestyleAlgo{boxruled}
\begin{algorithm}[H]
	\caption{\textsc{Matching Framework}}
	\label{alg:basicmatch}
	
	\textbf{At each edge $e = \{u,v\}$:}
	
	Set edge weight $w(e)\gets \frac{1}{2\Delta}$.
	Repeat every two rounds: 
	\begin{itemize}
		\item If $\sum_{e' \ni u} w(e') \le \frac{1}{2K}$ and $\sum_{e' \ni v} w(e') \le \frac{1}{2K}$, $w(e)\gets Kw(e)$.
		\item Self-nominate with probability $w(e)$.
		\item If $e$ self-nominates and no neighbouring edges do, \\ $e$ joins $M$, edge $e$ and all neighboring edges become\\ inactive (set $w(e) = 0$ permanently).
	\end{itemize}
\end{algorithm}

It can be shown that, if $K$ is set to be $\log^{\Theta(1)}\Delta$, this algorithm will (with moderately high probability, see Section \ref{sec:successprob}) give a $2+\eps$-approximate matching in $O(\frac{\log \Delta}{\log\log\Delta})$ rounds, for $\eps \ge 2^{-\log^{0.9} \Delta}$. We will not prove this here, since we are concerned with our adapted version achieving an $O(\frac{\log n}{\log^2\log n})$ complexity. We are not aware of this algorithm being explicitly previously published, but it seems likely that it was known by the time of \cite{BCGS17}; the algorithm of \cite{BCGS17} takes a similar approach but applied to the more difficult problem of obtaining a nearly-maximal independent set (from which results on matching are then derived), as opposed to Algorithm \ref{alg:basicmatch} which obtains a nearly-maximal matching. 

In our adapted algorithm (Algorithm \ref{alg:MM}), the main difference from Algorithm \ref{alg:basicmatch} is that the fractional matching will now be updated by clusters rather than individual edges (but the choice of self-nominating is still done individually by edges). Let $\const\ge 1000$ be a sufficiently large constant, used as the parameter in the decomposition algorithm of Lemma \ref{lem:localpart}, that will be chosen based on the value of $\delta$ in the statement of Theorem \ref{thm:MM}, and let $K:=\log^{99/\const} n<\log^{0.1} n$ be the value determining the speed of updating the fractional matching. The algorithm works by having nodes place `caps' on the total weight that \emph{clusters} are allowed to place on their adjacent edges, in order to ensure that the probability of having an adjacent edge attempt to join $M$ is sufficiently small. Clusters greedily assign weight to their member edges while respecting these caps. If a cluster meets the cap given to it by a node $v$, and that cap is still sufficiently below $1$, then in the next iteration the cap is increased by a factor of $K^2$. Otherwise, if the cluster does not meet the cap, then in the next iteration the cap is decreased by a factor of $K$. This process ensures that, for a sufficiently large number of rounds, any edge $\{u,v\}$ in the graph will have weight at least $\frac14 K^{-3}$ at one of its endpoints. Then, in such rounds, the random self-nomination step has an $\Omega(K^{-3})$ probability of removing $\{u,v\}$ (by putting it, or one of its neighboring edges, in the output matching). This is sufficient to obtain a $2+\eps$-approximate maximum matching.

\RestyleAlgo{boxruled}
\begin{algorithm}[H]
	\caption{\textsc{Approximate Maximum Matching}}
	\label{alg:MM}
	Apply the algorithm of Lemma \ref{lem:localpart} to the line graph of $G$, i.e. to cluster the \emph{edges} of the graph. For each cluster $C$, let $r_C$ denote its radius.
	
	\textbf{At each node $u$:}
	\begin{itemize}
		\item Initially set cap $\kappa_{u,C} = e^{\frac{-4\alpha}{1+r_{C}}}$, for all adjacent clusters $C$.
		\item Repeat: 
		\begin{enumerate}
			\item Send $\kappa_{u,C}$ to cluster $C$.
			\item Cluster $C$ will (in $2( r_{C}+1)$ rounds) return weights $w(\{u,v\})$ for all $\{u,v\} \in C$. Do:
			\begin{itemize}
				\item If $\sum_{\{u,v\} \in C }w(\{u,v\}) = \kappa_{u,C}$ and $\sum_{\text{clusters }C}\kappa_{u,C} \le \frac 14 K^{-3}$, $\kappa_{u,C} \gets K^2\kappa_{u,C}$.
				\item If $\sum_{\{u,v\} \in C }w(\{u,v\}) < \kappa_{u,C}$, $\kappa_{u,C} \gets K^{-1}\kappa_{u,C}$.
			\end{itemize}
		\end{enumerate}
	\end{itemize}
	
	\textbf{At each cluster $C$:}
	\begin{itemize}
		\item Repeat every $2( r_{C}+1)$ rounds: 
		\begin{enumerate}
			\item Collect all caps $\kappa_{u,C}$ to the cluster center.
			\item Cluster center first sets $w(\{u,v\}) \gets K^{-1} w(\{u,v\})$\\ for all active $\{u,v\}\in C$, then greedily increases each\\ such $w(\{u,v\})$ until either $\sum_{\{u,x\} \in C }w(\{u,x\}) = \kappa_{u,C}$\\ or $\sum_{\{v,x\} \in C }w(\{v,x\}) = \kappa_{v,C}$.
			\item Return new weights $w(\{u,v\})$ to nodes $u$ and $v$.
		\end{enumerate}
	\end{itemize}
	
	\textbf{At each edge $e$:}
	\begin{itemize}
		\item Repeat every two rounds: 
		\begin{enumerate}
			\item Self-nominate with probability $w(e)$.
			\item If $e$ self-nominates and no neighbouring edges do,\\ $e$ joins $M$, edge $e$ and all neighboring edges become inactive (set $w(e) = 0$ permanently).
		\end{enumerate}
	\end{itemize}
\end{algorithm}

\subsection{Analysis}
For an edge $\{u,v\}$, we analyze only the rounds in which $\{u,v\}$'s cluster $C$ receives its new caps $\kappa_{w,C}$. this occurs every $2( r_{C}+1)$ rounds. We call these rounds \emph{pivotal} rounds for $\{u,v\}$. We show that either one endpoint of $\{u,v\}$ already has a high total cap, or the product of its endpoints' caps for cluster $C$ increases by a factor of $K$. For the purposes of discriminating between the old and new caps, in the following lemma we denote by $\kappa_{u,C}$ the cap of node $u$ for cluster $u$ \emph{before} updating in the pivotal round, and we denote by $\kappa'_{u,C}$ the same cap \emph{after} updating.

\begin{lemma}\label{lem:pivotal}
	In each pivotal round for an active edge $\{u,v\}$, either 
	\begin{itemize}
		\item $\sum_{\text{clusters }C'} \kappa'_{u,C'} > \frac 14 K^{-3}$,
		\item $\sum_{\text{clusters }C'} \kappa'_{v,C'} > \frac 14 K^{-3}$,
		\item or $\kappa'_{u,C}\kappa'_{v,C} \ge  K\kappa_{u,C}\kappa_{v,C}$.
	\end{itemize}
\end{lemma}

\begin{proof}
	Since $\{u,v\}$ is active, when its cluster $C$ updates edge weights, the new weights will satisfy either $\sum_{\{u,x\} \in C }w(\{u,x\}) = \kappa_{u,C}$, or $\sum_{\{v,x\} \in C }w(\{v,x\}) = \kappa_{v,C}$: otherwise, $C$ could have further increased $w(\{u,v\})$. Assume without loss of generality that $\sum_{\{u,x\} \in C }w(\{u,x\}) = \kappa_{u,C}$. Then, when node $u$ updates its caps, either $\sum_{\text{clusters }C'} \kappa_{u,C'} > \frac 14 K^{-3}$ (i.e. the first point of the lemma is satisfied), or $u$ updates $\kappa'_{u,C} \gets K^2\kappa_{u,C}$. In the latter case, since $\kappa'_{v,C}\ge K^{_-1}\kappa_{v,C}$, we have that  $\kappa'_{u,C}\kappa'_{v,C} \ge  K\kappa_{u,C}\kappa_{v,C}$, i.e. the third point is satisfied. 
\end{proof}

We next show that the sum of caps around a node never gets too close to $1$.

\begin{lemma}\label{lem:capslow}
	For any node $v$, at any point in the algorithm, we have $\sum_{\text{clusters }C'} \kappa_{v,C'} \le \frac 14K^{-1}$.
\end{lemma}

\begin{proof}
	At the beginning of the algorithm, we have $\kappa_{v,C'} = e^{\frac{-4\alpha}{1+r_{C'}}}$. By Lemma \ref{lem:localpart}, the number of adjacent clusters to $v$ of radius $r_{C'}$ is at most $e^{\frac{3\alpha}{1+r_C}}$. So, the initial value of $\sum_{\text{clusters }C'} \kappa_{v,C'}$ is at most:
	\begin{align*}
	\sum_{r_{C'} = 0}^{\log^{1/3} n}e^{\frac{-4\alpha}{1+r_{{C'}}}} \cdot e^{\frac{3\alpha}{1+r_{C'}}}
	&\le \sum_{r_{C'} = 0}^{\log^{1/3} n}e^{\frac{-\alpha}{1+r_{{C'}}}}\\
	&\le \sum_{r_{C'} = 0}^{\log^{1/3} n}e^{\frac{-\alpha}{2\log^{1/3} n}}\\
	&\le (\log^{1/3}n+1) e^{-8\log^{0.6}n}\\
	&\le \frac14K^{-3}\enspace.
	\end{align*}
	
	Here we used throughout that $n$ is at least a sufficiently large constant.
	
	Subsequently in Algorithm \ref{alg:MM}, caps $\kappa_{v,C'}$ are only raised (by factor $K^2$) if $\sum_{\text{clusters }C'}\kappa_{v,C'} \le \frac 14 K^{-3}$, so we maintain the property $\sum_{\text{clusters }C'}\kappa_{v,C'} \le \frac 14K^{-1}$.
\end{proof}

An upper bound on the sum of caps around a node implies an upper bound on the sum of weights around the node:

\begin{lemma}\label{lem:weightslow}
	For any node $v$, at any point in the algorithm, we have $\sum_{\{u,v\}\in E} w(\{u,v\}) \le \frac 14$.
\end{lemma}

\begin{proof}
	By Lemma \ref{lem:capslow}, $\sum_{\text{clusters }C'} \kappa_{v,C'} \le \frac 14K^{-1}$. For any cluster $C'$,  $\sum_{\{u,v\} \in C' }w(\{u,v\}) \le K \kappa_{v,{C'}}$; this is since, when weights are updated, they satisfy $\sum_{\{u,v\} \in C' }w(\{u,v\}) \le \kappa_{v,{C'}}$, and when caps are updated, they are reduced by at most factor $K$. So, summing over all adjacent clusters to $v$, $\sum_{\{u,v\}\in E} w(\{u,v\}) \le \frac 14$. 
\end{proof}

Now we are ready to show the main technical result, a lower bound on the probability that an edge is removed from the graph over the course of the algorithm.

\begin{lemma}\label{lem:edgeremoved}
	For any constant $\delta>0$, there exists constant $\const$ such that if Algorithm \ref{alg:MM} is run for $\frac{\const^2\log n}{\log^2 \log n}$ rounds (not including the time of performing \textsc{Partition}), any edge $\{u,v\}$ is removed (becomes inactive) with probability at least $1-e^{- \Omega(\log^{1-\delta} n)}$.
\end{lemma}

\begin{proof}
	Edge $\{u,v\}$ has at least $P:=\frac{\const^2\log n}{2(1+r_C) \log^2 \log n}$ pivotal rounds, where $C$ is the cluster of $\{u,v\}$. We bound the number of these pivotal rounds in which $\kappa_{u,C}\kappa_{v,C}$ increases by a factor of at least $K$. Denote this number $c$. Since otherwise $\kappa_{u,C}\kappa_{v,C}$ decreases by a factor of at most $K^2$, and by Lemma \ref{lem:capslow} $\kappa_{u,C}\kappa_{v,C}\le \frac{1}{16} K^{-2}$, we have \[K^{c}\cdot K^{-2(P-c)}\cdot e^{\frac{-8\alpha}{1+r_C}}\le \frac{1}{16} K^{-2}\enspace.\]
	
	This gives:
	
	\begin{align*}
	c\ln K - 2(P-c)\ln K - \frac{8\alpha}{1+r_C} &\le \ln (\frac {1}{16}K^{-2})\\
	(3c-2P)\ln K - \frac{8\alpha}{1+r_C} &\le 0\\
	3c-2P  &\le \frac{8\alpha}{(1+r_C)\ln K}  \\
	c  &\le \frac{\frac{8\alpha}{(1+r_C)\ln K} +2P}{3}
	\end{align*}
	
	Plugging in $\alpha = \frac{\const\log n}{\log\log n}$, $P=\frac{\const^2\log n}{2(1+r_C) \log^2 \log n}$, and $K=\log^{99/\const}n$, 
	
	\begin{align*}
	c  &\le \frac{\frac{8\const^2\log n}{99(1+r_C)\ln \log n\cdot \log\log n}+\frac{\const^2\log n}{(1+r_C) \log^2 \log n} }{3} \\
	&\le\frac{\frac{8}{99\ln 2}+1}{3}\cdot \frac{\const^2\log n}{(1+r_C) \log^2 \log n}\\
	&< \frac{0.38\const^2\log n}{(1+r_C) \log^2 \log n}\enspace.
	\end{align*}
	
	So, there are at least $P-\frac{0.38\const^2\log n}{(1+r_C) \log^2 \log n}\ge\frac{0.12\const^2\log n}{(1+r_C) \log^2 \log n} \ge \frac{0.12\const^2\log n}{(1+\log^{3/\const} n) \log^2 \log n} \ge \frac{\const^2\log^{1-3/\const} n}{9\log^2 \log n}$ pivotal rounds in which $\kappa_{u,C}\kappa_{v,C}$ does not increase by a factor of at least $K$. By Lemma \ref{lem:pivotal}, in these rounds either $\sum_{\text{clusters }C'} \kappa_{u,C'} > \frac 14 K^{-3}$ or $\sum_{\text{clusters }C'} \kappa_{v,C'} > \frac 14 K^{-3}$, so $\sum_{\text{clusters }C'} (\kappa_{u,C'}+\kappa_{v,C'}) > \frac 14 K^{-3}$.
	
	Furthermore, we can choose at least $\frac{\const^2\log^{1-6/\const} n}{18\log^2 \log n}$ such pivotal rounds that are each at least $2\log^{3/\const}n$ rounds apart from each other. Denote by $S$ the set of such rounds, and denote by $\kappa^i_{v,C'}$ the value of $\kappa_{v,C'}$ at round $i$. Then, we have $\sum_{i\in S}\sum_{\text{clusters }C'} (\kappa^i_{u,C'}+\kappa^i_{v,C'}) > \frac 14 K^{-3} |S| \ge \frac{\const^2\log^{1-6/\const} n}{72K^{3}\log^2 \log n}$. 
	
	\hide{Suppose W.L.O.G. that $$\sum_{i\in S}\sum_{\text{clusters }C'} \kappa^i_{v,C'} > 2 K^{-3}\log^{0.35} n\enspace.$$}
	
	Now, we need to bound the \emph{weights} contributed by each cluster around a node $v$, rather than just the caps. Suppose that cluster $C'$ meets its cap for node $v$ and round $i$, and the next time it does so is at round $j$. That is,
	
	\begin{itemize}
		\item $\sum_{\{v,x\} \in C' }w^i(\{v,x\}) = \kappa^i_{v,C'} $;
		\item $\sum_{\{v,x\} \in C' }w^{j}(\{v,x\}) = \kappa^j_{v,C'} $; and
		\item for all $i'\in [i+1, j-1]$, $\sum_{\{v,x\} \in C' }w^{i'}(\{v,x\}) < \kappa^{i'}_{v,C'} $.
	\end{itemize}
	
	Then, observe how $\kappa^i_{v,C'}$ changes between rounds $i$ and $j$: we have $\kappa^{i'}_{v,C'} \le \kappa^i_{v,C'} K^2$ for all $i'\in [i+1, j-1]$, since $\kappa_{v,C'}$ may increase by a factor of $K^2$ in round $i$ but will henceforth not increase again until round $j$. We also have that every $2r_{C'}\le 2\log^{0.3}n$ rounds during the interval $[i+1, j-1]$, $\kappa_{v,C'}$ will decrease by a $K^{-1}$ factor. So, 
	
	\[\sum_{i' \in [i, j-1] \cap S}\kappa^i_{v,C'} \le \sum_{t=0}^{\infty} \kappa^i_{v,C'} K^{2-t} \le 2K^2\kappa^i_{v,C'} = 2K^2\sum_{\{v,x\} \in C' }w^i(\{v,x\})\enspace.\] 
	
	Summing over all such intervals $[i,j-1]$, we get:
	\[\sum_{i' \in S}\kappa^i_{v,C'} \le 2K^2\sum_{\{v,x\} \in C' }\sum_{i\in [\frac{\const^2\log n}{\log^2 \log n}]}w^i(\{v,x\})\enspace.\]
	
	So, 
	
	\begin{align*}
	\sum_{\{v,x\} \in E(v)}\sum_{i\in [\frac{\const^2\log n}{\log^2 \log n}]}w^i(\{v,x\})
	&= \sum_{\text{clusters $C'$}}\sum_{\{v,x\} \in C' }\sum_{i\in [\frac{\const^2\log n}{\log^2 \log n}]}w^i(\{v,x\})\\
	&\ge \frac12 K^{-2} \sum_{\text{clusters $C'$}}\sum_{i' \in S}\kappa^i_{v,C'}
	\enspace.
	\end{align*}
	
	Summing over both endpoints of edge $\{u,v\}$:
	\begin{align*}
	\sum_{i\in [\frac{\const^2\log n}{\log^2 \log n}]}\left( \sum_{\{v,x\} \in E(v)}w^i(\{v,x\}) + \sum_{\{u,x\} \in E(u)}w^i(\{u,x\}) \right)
	&\ge
	\frac12 K^{-2} \sum_{\text{clusters $C'$}}\sum_{i' \in S}\kappa^i_{v,C'} + \kappa^i_{u,C'}\\
	&\ge \frac{\const^2\log^{1-6/\const} n}{144K^{5}\log^2 \log n}\enspace.
	\end{align*}
	
	Without loss of generality, assume that $v$ is the endpoint of edge $\{u,v\}$ with higher total weight in round $i$, i.e.  $$\sum_{\{v,x\} \in E(v)}w^i(\{v,x\}) \ge \sum_{\{u,x\} \in E(u)}w^i(\{u,x\})\enspace.$$ Then, the probability that edge $\{u,v\}$ becomes inactive in round $i$ is at least
	
	\begin{align*}
	&\sum_{\{v,x\} \in E(v)} \Prob{\text{$\{v,x\}$ self-nominates and no neighbbouring edge does}}\\
	&\hspace{1in}\ge \sum_{\{v,x\} \in E(v)} w(\{v,x\}) \cdot \prod_{\{v,y\} \in E(v)}\left(1-w(\{v,y\})\right)\cdot \prod_{\{x,y\} \in E(x)}\left(1-w(\{x,y\})\right)\\
	&\hspace{1in}\ge \sum_{\{v,x\} \in E(v)} w(\{v,x\}) \cdot 4^{-\prod_{\{v,y\} \in E(v)}w(\{v,y\})}  \cdot 4^{-\prod_{\{x,y\} \in E(x)}w(\{x,y\})} \\
	&\hspace{1in}\ge \sum_{\{v,x\} \in E(v)} w(\{v,x\}) \cdot 4^{-\frac14}  \cdot 4^{-\frac14} \\
	&\hspace{1in}= \frac 12\sum_{\{v,x\} \in E(v)} w(\{v,x\})  \\
	&\hspace{1in}\ge \frac 14\left( \sum_{\{v,x\} \in E(v)}w^i(\{v,x\}) + \sum_{\{u,x\} \in E(u)}w^i(\{u,x\}) \right)\enspace. \\
	\end{align*}
	Here we used the inequality $1-t \ge 4^{-t}$ for $t \in [0, \frac12]$.
	
	Furthermore, this probability bound holds regardless of the random choices in any other round. So, the probability that edge $\{u,v\}$ is not inactive after $\frac{\const^2\log n}{\log^2 \log n}$ rounds is at most 
	
	\begin{align*}
	&	\prod_{i\in [\frac{\const^2\log n}{\log^2 \log n}]}\left(1- \frac 14\left( \sum_{\{v,x\} \in E(v)}w^i(\{v,x\}) + \sum_{\{u,x\} \in E(u)}w^i(\{u,x\}) \right)\right)\\
	&\hspace{1in}\le
	\prod_{i\in [\frac{\const^2\log n}{\log^2 \log n}]}e^{- \frac 14\left( \sum_{\{v,x\} \in E(v)}w^i(\{v,x\}) + \sum_{\{u,x\} \in E(u)}w^i(\{u,x\}) \right)}\\
	&\hspace{1in}= e^{-\frac14\sum_{i\in [\frac{\const^2\log n}{\log^2 \log n}]}\left( \sum_{\{v,x\} \in E(v)}w^i(\{v,x\}) + \sum_{\{u,x\} \in E(u)}w^i(\{u,x\}) \right)}\\
	&\hspace{1in}\le e^{- \frac{\const^2\log^{1-6/\const} n}{576K^{5}\log^2 \log n}}\\
	&\hspace{1in}\le e^{- \frac{\const^2\log^{1-501/\const} n}{576\log^2 \log n}}\\
	&\hspace{1in}\le e^{- \Omega(\log^{1-\delta} n)}\enspace.
	\end{align*}
	
	letting $\const:=502/\delta$. 
\end{proof}

It remains to prove the approximation ratio of the resulting matching:

\begin{proof}[Proof of Theorem \ref{thm:MM}]
	By Lemma \ref{lem:edgeremoved}, after $O(\frac{\log n}{\log^2 \log n})$ rounds each edge has a matched endpoint with probability at least $1-e^{- \Omega(\log^{1-\delta} n)}$. In particular, each edge in the (canonical) optimum maximum matching $OPT$ remains active with probability at most $e^{- \Omega(\log^{1-\delta} n)}$. Denote by $\lambda$ the number of edges in $OPT$ that remain active. Then, $\Exp{\lambda} \le |OPT|e^{- \Omega(\log^{1-\delta} n)}$. By Markov's inequality, $\Prob{\lambda\ge \frac{\eps}{3} |OPT|}\le \frac{3}{\eps} \cdot e^{- \Omega(\log^{1-\delta} n)} =\frac{1}{\eps} e^{- \Omega(\log^{1-\delta} n)}$. 
	
	Since each edge placed in the matching $M$ by Algorithm \ref{alg:MM} renders at most two edges in $OPT$ inactive, with probability at least $1-\frac{1}{\eps} e^{- \Omega(\log^{1-\delta} n)}$ we have 
	
	\[|M| \ge \frac{1}{2}(|OPT| - \lambda) \ge \frac{1}{2}(|OPT| - \frac{\eps}{3} |OPT|) =  \frac{3-\eps}{6}|OPT| \ge \frac{|OPT|}{2+\eps}\enspace. \]
\end{proof}

\subsection{Discussion of success probability}\label{sec:successprob}

We note that, while this probability of success in Theorem \ref{thm:MM} is slightly weaker than the standard definition of `with high probability' (i.e. probability $1-n^{-c}$ for some $c\ge 1$), this is not unprecedented for approximate matching algorithms, particularly those taking $o(\log n)$ rounds. The main algorithm of \cite{BCGS17}, for example, obtains a bound of $2^{-\log^{0.7} n}$ on the probability that any particular edge remains active, similar to (but slightly weaker than) our bound of $e^{- \Omega(\log^{1-\delta} n)}$ here. While \cite{BCGS17} claims a `with high probability' bound, and sketches an argument that the number of remaining active edges is concentrated around its expectation (that would also apply to our algorithm if true), this argument fails for $\Delta\gtrapprox n^{1/3}$. The preprint version of \cite{BCGS17} gives an alternative algorithm that succeeds with probability $e^{- \Omega(\eps|OPT|)}$, which is stronger in many cases but still is not always w.h.p. under the standard definition.

\hide{
	\subsubsection{Concentration Bound}

	\begin{proof}
		From proof of Lemma \ref{lem:edgeremoved}, we have that for any edge $\{u,v\}$, at least one of its endpoints (W.L.O.G. named $v$) satisfies:
		
		\begin{align*}
		\sum_{\{v,x\} \in E(v)}\sum_{i\in [\frac{9000\log n}{\log^2\log n}]}w^i(\{v,x\})\ge
		4 \log^{0.3} n\enspace.
		\end{align*}
		
		This event happens with certainty, irrespective of the random choices during the algorithm. Let $T$ be some subset of the canonical maximum matching $OPT$ of size $\log^{1.7} n$. Summing over all edges in $T$, 
		
		\begin{align*}
		\sum_{i\in [\frac{9000\log n}{\log^2\log n}]}\sum_{\{u,v\} \in T}\sum_{\{v,x\} \in E(v)}w^i(\{v,x\})\ge
		4 \log^{2} n\enspace.
		\end{align*}
		
		We bound the probability that all edges in $T$ remain active at the end of the algorithm. In each round $i$, let $f_i$ denote the number of edges $\{u,v\}$ for which node $v$ is matched in round $i$. 
		
		VVV THe problem is that choices of edges distance $2$ away are also relevant, and they increase the McD bound too much.

		So, there must be at least 
\end{proof}}

\section{Distributed Approximate Minimum Weighted Vertex Cover}
We now give our algorithm for approximate minimum weighted vertex cover, which is inspired by the algorithms of \cite{BCS17} and \cite{BEKS18}, but requires significant changes to incorporate the graph decomposition.

\subsection{Algorithm Background}

\newcommand{\Amount}{t}
\newcommand{\Deal}{request}
\newcommand{\Vault}{vault}
\newcommand{\Bank}{bank}
\newcommand{\Budget}{budget}
\newcommand{\InCover}{\texttt{InCover}}
\newcommand{\NotInCover}{\texttt{NotInCover}}
\newcommand{\ceil}[1]{\left\lceil #1 \right\rceil}
\newcommand{\floor}[1]{\left\lfloor #1 \right\rfloor}

\newcommand{\parentheses}[1]{\left(#1\right)}
\newcommand{\logp}[1]{\log\parentheses{#1}}
\newcommand{\set}[1]{\left\{#1\right\}}

\newcommand{\Dv}{\ensuremath{\Delta}}
\newcommand{\Kv}{\ensuremath{K}}

\newcommand{\maxr}{\ensuremath{ \log^{0.1} n}}

The algorithm follows the local ratio template also used by \cite{LPP15, BCS17,BEKS18} (from which the following paragraph is adapted, but with notation altered to match the algorithms): 

Assume a given weighted graph $G = (V, w, E)$, where $w_0 : V \rightarrow \mathbb R^+$ is an assignment of weights for the vertices\footnote{The problem of approximate minimum weighted vertex cover only makes sense with non-negative weights - if negative weights are allowed, then an isolated pair of nodes can be added to any instance in order to set the minimum cover's weight to $0$, in which case any approximate minimum cover is an exact minimum cover.}. Let $\delta : E \rightarrow  \mathbb R^+$ be a function that assigns weights to edges. We say that $\delta$ is $G$-valid if for every $v\in V $, $\sum_{e:v\in e} \delta(e)\le w_0(v)$, that is, the sum of weights of edges that touch a vertex is at most the weight of that vertex in $G$. Fix any G-valid function $\delta$. Define $w_{\delta}: V \rightarrow \mathbb R^+$ by $w_{\delta}=\sum_{e: v\in e}\delta(e)$, and let $w: V \rightarrow \mathbb R^+$ such that $w(v) = w_0(v)- w_{\delta}(v)$. Let $S_\delta = \{v \in V : w(v) \le \eps' w_0(v)\}$, where $\eps'= \eps/(2 + \eps)$.

\begin{theorem}[Theorem 2.1 of \cite{BCS17}]\label{thm:localratio}
	Fix $\eps > 0$ and let $\delta$ be a G-valid function. Let $OPT$ be the sum of weights
	of vertices in a minimum weight vertex cover $S_{OPT}$ of G. Then $\sum_{v\in S_\delta}w(v)\le (2+\eps)OPT$.
	In particular, if $S_\delta$ is a vertex cover, then it is a $(2 + \eps)$-approximation for MWVC for G.
\end{theorem}

This framework is converted into an algorithm in \cite{BCS17} as follows: we start off with $\delta$ being uniformly $0$, and $w=w_0$ being the weight function given as input for the MWVC instance. Then, we iteratively assign some weight to edges, i.e. increase $\delta(e)$ for some $e$. This in turn raises $w_\delta$ and reduces $w$ for the endpoints of that edge by the same amount (though in the algorithm of \cite{BCS17} and our own, only the change to $w$ is explicitly kept track of).  So long as we ensure that 
\begin{enumerate}
	\item all $w(v)$ values remain non-negative, indicating that $\delta$ is G-valid, 
	\item we only add nodes $v$ to the output set when $w(v)\le \eps' w_0(v)$, and
	\item the output set eventually forms a vertex cover,
\end{enumerate}

then by Theorem \ref{thm:localratio}, this output set is indeed a $(2 + \eps)$-approximate MWVC. 

An algorithm under this framework, therefore, need only specify which edges to assign weight to, and how much weight to assign, while respecting the above conditions. The algorithm of \cite{BCS17} does so by dividing nodes' current weight $w(v)$ into a \emph{bank} and a \emph{vault}. Nodes send requests along their adjacent edges to (implicitly) assign weight from their vault along that edge, and accordingly reduce the weight at each endpoint by the same amount. Requests are answered by the other endpoint $u$ of the edge by assigning weight from $u$'s bank. We follow the same approach, though we do not use the vault and bank terminology since our introduction of the graph decomposition makes this division less clear.

We note that \cite{BEKS18} generalized the algorithm of \cite{BCS17} to improve the dependency on $\eps$ to logarithmic; our algorithm is closer to that of \cite{BCS17}, since that of \cite{BEKS18} is harder to adapt directly, but we get the same improvement in $\eps$-dependency via a process that essentially repeats the algorithm of \cite{BCS17} for each of the $1+O(\frac{\log 1/\eps}{\log\log n})$ \emph{levels} in the algorithm of \cite{BEKS18}.

\subsection{Algorithm Overview}

Our algorithm proceeds in $1+O(\frac{\log 1/\eps}{\log\log n})$ \emph{stages}, comparable to the levels in \cite{BEKS18} (though in the algorithm of \cite{BEKS18} the levels are handled concurrently, whereas here we handle them sequentially). The goal in each stage is to reduce all non-terminated nodes' weight by a factor of $\log^{\Theta(1)} n$.

Throughout the algorithm, nodes can be in one of four states: active, inactive, $\InCover$, and $\NotInCover$. $\InCover$, and $\NotInCover$ are terminal states, and indicate that the vertex is in and out of the final vertex cover respectively; once nodes reach those states, they will inform their neighbors and cease participation.

Each stage consists of two phases. In phase $1$, all non-terminated nodes start as active, and some will become inactive during the phase. No nodes terminate during phase $1$. Active nodes $u$ send requests asking to increase weight on their adjacent edges (and therefore correspondingly reduce their own weight $w(u)$). Requests are sent by individual nodes, but are answered by the recipient's cluster. Because requests are answered with a delay (based on the diameter of the recipient's cluster), weight sent in requests is ``locked up'' until the request is answered, and therefore nodes do not know exactly how much weight they have available. Because of this, when clusters fulfil requests, they must leave sufficient weight aside to cover all requests that $v$ might send before the cluster next receives updated weights. Nodes become inactive once their weight reduces by a $\log^{0.1}n$ factor, compared to their weight at the start of the phase. The goal of phase $1$ is to ensure that there are no edges between active nodes, i.e. at least one endpoint of every edge becomes inactive.

In phase $2$, active nodes again send requests. Since there are no edges between active nodes, requests are only received by inactive nodes. Again, requests are answered by clusters, but since they are answered only on behalf of inactive nodes, no weight needs to be held back to account for weight locked up in sent requests. So, the cluster can use the entire remaining weight of nodes to answer requests, and knows the true current value of $w(v)$ for all inactive nodes $v$ it contains (since such weight values are only updated by the cluster itself). 

During phase $2$, no nodes become inactive, but nodes $v$ can terminate in one of two ways: if $v$'s weight reaches the target value of $\eps' w_0(v)$, then $v$ can join the vertex cover. If $v$ has no non-terminated neighbors, it outputs $\NotInCover$. Note that in this case, all of $v$'s neighbors must have already output $\InCover$, since they could not have output $\NotInCover$ while $v$ was still non-terminated.

The goal of phase $2$ is to ensure than no active nodes remain, i.e. all active nodes will terminate. This means that at the end of a stage, the only non-terminated nodes are inactive, and therefore their weights have decreased by at least a $\log^{0.1}n $ factor since the start of the stage. Inactive nodes are set to be active again for the start of the next stage. 

Since each stage reduces weights by a $\log^{0.1}n $ factor, after $1+ \log_{\log^{0.1}n }\frac{1}{\eps'} = 1+O(\frac{\log 1/\eps}{\log\log n})$ stages all remaining nodes $v$ must have weight below $\eps'w_0(v)$, and so must terminate.

\subsection{Notation}
Our algorithm uses the following notation:

\begin{itemize}
	\item $w_0(v)$ is the starting, input weight of node $v$.
	\item $w(v)$ is initialized to $w_0(v)$, but is decreased over the course of the algorithm.
	\item $w_i(v)$ is the value of $w(v)$ at the beginning of stage $i$. $w_i'(v)$ is the value of $w(v)$ at the beginning of phase $2$ of stage $i$.
	\item $w^t(v)$ is the value of $w(v)$ in round $t$ of the current phase. It is only used in the following way: clusters $C$ operate in periods of length $2(r_C+1)$, where $r_C$ is the radius of $C$. In each period, starting at round $2k(r_C+1)$ for some integer $k$, when the cluster wishes to use $w(v)$ it only knows the value at the start of the period, i.e. as it was at round $2k(r_C+1)$. Therefore, the notation $w^{2k(r_C+1)}$ is used to distinguish from the \emph{current} value $w(v)$, which may have changed since.
	\item $NC_{r}(v):=\{\text{clusters $C$ of radius $r$ containing active neighbors of $v$}\}$, and $d_r(v):= |NC_{r}(v)|$. These change over time as nodes become inactive or terminated, and we assume that they are updated by $v$ in every round.
	\item Similarly	$NC'_{r}(v):=\{\text{clusters $C$ of radius $r$ containing non-terminated neighbors of $v$}\}$, and $d'_r(v) := |NC'_{r}(v)|$. Again, these change as nodes terminate, and we assume that they are updated by $v$ each round.
	\item $E_C := \{(u,v): \{u,v\}\in E \land v \in C\}$, i.e. the set of edges adjacent to nodes in $C$, directed so that edges with both endpoints in $C$ appear twice. 
\end{itemize}

\begin{algorithm}[htbp]
	\SetKwProg{Fn}{}{ }{}
	\caption{Approximate MWVC, Overall Algorithm}
	\label{alg:MWVC}
	Apply the algorithm of Lemma \ref{lem:localpart} with $\const = 30$
	
	\For{$i \in [1+O(\frac{\log 1/\eps}{\log\log n})]$, in Stage $i$,}{
		All non-terminated nodes are set to `active'
		
		Run Phase 1 (Algorithm \ref{alg}) for $O(\frac{\log n}{\log^2\log n})$ rounds

		Run Phase 2 (Algorithm \ref{alg2}) for $O(\frac{\log n}{\log^2\log n})$ rounds}
\end{algorithm}

\begin{algorithm}[h]
	\SetKwProg{Fn}{}{ }{}
	\caption{Approximate MWVC, Phase 1 of Stage $i$}
	\label{alg}
	\textbf{At each active node $v$:}\Fn{}{		
		\ForEach{$r\in [\maxr]$, every $2(r+1)$ rounds}
		{	
			\ForEach{cluster $C \in NC_r(v)$}
			{		
				$\Deal(v,C) \gets  \frac{w_i(v)}{ d_r(v)\log^{0.3} n}$\\
				Send $\Deal(v,C)$ to cluster $C$\\
				Within $2(r+1)$ rounds, $C$ responds with $\Budget(v,C)$\\
				$w(v) \gets  w(v) - \Budget(v,C)$\\			
			}
		}
		\textbf{In every round,} \If{$w(v) \le \log^{-0.1}n \cdot w_i(v)$}
		{
			$v$ becomes inactive, informs neighbors
		}
	}
	\textbf{At each cluster $C$, every $2(r_C+1)$ rounds, starting at round $2k(r_C+1)$:} \Fn{}{	
		
		Initially set $\Budget(u,v) \gets  0$ for all $(u,v) \in E_C$.
		
		\ForEach{$(u,v) \in E_C$, in any order}
		{
			$\Budget(u,v) \gets \min\{\Deal(u,C) - \sum\limits_{(u,w) \in E_C}{\Budget(u,w)} , w^{2k(r_C+1)}(v)-\log^{-0.1}n \cdot w_i(v)- \sum\limits_{(w,v) \in E_C}{\Budget(w,v)} \}$\\
			
		}
		\ForEach{$u\in N(C)$}
		{
			Send $\Budget(u,C) \gets \sum\limits_{(u,w) \in E_C}{\Budget(u,w)} $ to $u$\\
			
		}
		\ForEach{$v\in C$}
		{
			$w(v) \gets w(v) - \sum\limits_{(w,v) \in E_C}{\Budget(w,v)} $\\
		}
	}
	
\end{algorithm}

\begin{algorithm}[h]
	\SetKwProg{Fn}{}{ }{}
	\caption{Approximate MWVC, Phase 2 of Stage $i$}
	\label{alg2}	
	\textbf{At each active node $v$:}\Fn{}{		
		\ForEach{$r\in [\maxr]$, every $2(r+1)$ rounds}
		{		
			\ForEach{cluster $C \in NC_r(v)$}
			{		
				$\Deal(v,C) \gets  \frac{w'_i(v)}{ d'_r(v)\log^{0.2}n}$\\
				Send $\Deal(v,C)$ to cluster $C$\\
				Within $2(r+1)$ rounds, $C$ responds with $\Budget(v,C)$\\
				$w(v) \gets  w(v) - \Budget(v,C)$\\			
			}
		}
		
	}
	\textbf{At each node $v$, every round:}\Fn{}{		
		\If{$w(v) \le \eps' w_0(v)$}
		{
			$v$ outputs $\InCover$, informs neighbors, terminates
		}
		\If{$v$ has no non-terminated neighbors}
		{
			$v$ outputs $\NotInCover$, informs neighbors, terminates
		}
	}
	\textbf{At each cluster $C$, every $2(r_C+1)$ rounds, starting at round $2k(r_C+1)$:} \Fn{}{	
		
		Initially set $\Budget(u,v) \gets  0$ for all $(u,v) \in E_C$.
		
		\ForEach{$(u,v) \in E_C$, in any order}
		{
			$\Budget(u,v) \gets \min\{\Deal(u,C) - \sum\limits_{(u,w) \in E_C}{\Budget(u,w)} , w^{2k(r_C+1)}(v)- \sum\limits_{(w,v) \in E_C}{\Budget(w,v)} \}$\\
			
		}
		\ForEach{$u\in N(C)$}
		{
			Send $\Budget(u,C) \gets \sum\limits_{(u,w) \in E_C}{\Budget(u,w)} $ to $u$\\
			
		}
		\ForEach{$v\in C$}
		{
			$w(v) \gets w(v) - \sum\limits_{(w,v) \in E_C}{\Budget(w,v)} $\\
		}
	}

\end{algorithm}

\subsection{Analysis}
Algorithm \ref{alg:MWVC} initially applies $\textsc{Partition}(F)$ (Algorithm \ref{alg:part}) with $F(x) =  1-(1+x)^{-\alpha}$, $\alpha = \frac{30\log n}{\log\log n}$. By Lemma \ref{lem:localpart}, therefore, all clusters are of radius at most $\log^{0.1} n$, and every node has at most $e^{\frac{90\log n}{(q+1)\log\log n}}$ adjacent clusters of radius $q$, with high probability. We now analyze progress during the algorithm's phases.

\paragraph{Analysis of Phase 1}

The first property to check is that all node weights remain non-negative, which is non-trivial due to the asynchronicity introduced by the clustering.

\begin{lemma}\label{lem:posweights}
	For any node $v$, during any Phase 1, $w(v)\ge 0$.
\end{lemma}

\begin{proof}
	Assume that $w(v)$ drops below $\log^{-0.1}n \cdot w_i(v)$ in some round $t'$ in Phase $1$ of some stage $i$. Let $k$ be the integer such that $t\in [2k(r_C+1), 2(k+1)(r_C+1))$, where $C$ is $v$'s own cluster - that is, $w(v)$ drops below $\log^{-0.1}n  \cdot w_i(v)$ during the $k^{th}$ period of $v$'s cluster's loop. We analyze how $w(v)$ decreases during this period:
	
	In resolving requests received by $v$, cluster $C$ reduces $w(v)$ by at most the value of $ w^{2k(r_C+1)}(v)-\log^{-0.1}n \cdot w_i(v)$ during the period. The amount that $w(v)$ reduces due to requests sent by $v$ during the period is at most $\sum_{r \in [\maxr]} \frac{ w_i(v)}{\log^{0.3} n} \cdot \frac{2(r+1)}{2(r_C+1)} \le \log^{-0.1}n \cdot w_i(v)$. Therefore $w^{2(k+1)(r_C+1)-1}(v)\ge w^{2k(r_C+1)}(v)-(w^{2k(r_C+1)}(v) -  \log^{-0.1}n  \cdot w_i(v)) -  \log^{-0.1}n  \cdot w_i(v) = 0$.
	
	Since $v$ becomes inactive after the period, further updates to $w(v)$ are made only by its cluster, which will never reduce $w(v)$ below $0$.
\end{proof}

The next lemma characterizes the progress of an active node $v$: whenever it receives \Budget\ values from clusters of radius $r$ (i.e. every $2(r+1)$ rounds), either $v$'s weight is decreased by a sufficient addive amount, or its number of such clusters containing active neighbors is reduced by a sufficient multiplicative amount.

\begin{lemma}\label{lem:progressP1a}
	For any active node $v$, when clusters of radius $r$ return \Budget\ values to a node $v$, either
	\begin{itemize}
		\item $w(v)$ is reduced by at least $\frac{ w_i(v)}{\log^{0.7} n}$, or
		\item $d_r(v)$ reduces by at least a multiplicative factor of $\log^{-0.4} n$.
	\end{itemize}
\end{lemma}

\begin{proof}
	Fix any $r\in [\maxr]$, and consider any particular period of $2(r+1)$ rounds, where node $v$ sends $\Deal$\ values to clusters in $NC_r(v)$ at the start of the period, and receives \Budget\ values and updates its weight at the end of the period. Let $w(v)$, $d_r(v)$ denote the values at the start of the period, and $\hat w(v)$, $\hat d_r(v)$ denote the updated values at the end of the period.
	
	Assume that $w(v)$ is not reduced by at least $\frac{ w_i(v) }{\log^{0.7} n}$, i.e. $\hat w(v)> w(v) - \frac{ w_i(v)}{\log^{0.7} n}$. Then, there are at most $d_r(v)\log^{-0.4} n$ clusters $C\in NC_r(v)$ for which $\Budget(v,C)=\Deal(v,C)$. In all other clusters $C$, for all $(v,w)\in E_C$, $w$ is left with at most the $\log^{-0.1}n \cdot w_i(w)$ weight that is not used by clusters to answer requests, and therefore all such $w$ become inactive and $C$ leaves $NC_r(v)$. So, $\hat d_r(v) \le d_r(v)\log^{-0.4} n$.
\end{proof}

Now, we can show that after $O(\frac{\log n}{\log^2\log n})$ rounds, no edges remain between active nodes, which was the goal of phase 1.

\begin{lemma}\label{lem:progressP1b}
	After $O(\frac{\log n}{\log^2\log n})$ rounds of Phase 1, all nodes $v$ that are still active have no active neighbors.
\end{lemma}

\begin{proof}
	For any fixed node $v$ and $r\in [\maxr]$, $d_r(v)$ is initially at most $e^{\frac{90\log n}{(r+1)\log\log n}}$, by the properties of the clustering. In each of the $\Theta(\frac{\log n}{(r+1)\log^2\log n})$ periods of $2(r+1)$ rounds within the phase, either $w(v)$ is reduced by at least $\frac{ w_i(v)}{\log^{0.7} n}$, or $d_r(v)$ reduces by a multiplicative factor of at least $\log^{-0.4} n$. Therefore, choosing a sufficiently large constant in the $O(\frac{\log n}{\log^2\log n})$ running time of the phase, either $v$ becomes inactive during the phase, $w(v)$ falls below $w_i(v)-\frac{ w_i(v) }{\log^{0.7} n}\cdot \frac{\log n}{(r+1)\log^2\log n}< 0$ (which is not possible), or $d_r(v) \le e^{\frac{90\log n}{(r+1)\log\log n}} \cdot \left(\log^{-0.4}n\right)^{\frac{1000\log n}{(r+1)\log^2\log n}}<1$, i.e. $d_r(v)=0$. Since this holds for all $r\in [\maxr]$, if $v$ remains active it must have no remaining active neighbors. 
\end{proof}

\paragraph{Analysis of Phase 2}

Phase 2 is analyzed similarly. During phase 2 it is clear that no $w(v)$ values can drop below $0$, since active nodes only reduce $w(v)$ due to requests they send, and inactive nodes only reduce $w(v)$ due to requests their clusters receive. So, we proceed to a progress lemma similar to Lemma \ref{lem:progressP1a}; the difference is that we are now concerned with the number of clusters containing non-terminated, rather than active, neighbors. (Note that for an active node $v$, during phase $2$, its non-terminated neighbors must all be passive by Lemma \ref{lem:progressP1b}.)

\begin{lemma}\label{lem:progressP2}
	During phase 2, for any active node $v$, when clusters of radius $r$ return \Budget\ values to a node $v$, either
	\begin{itemize}
		\item $w(v)$ is reduced by at least $\frac{ w_i(v)}{\log^{0.7} n}$, or
		\item $d'_r(v)$ reduces by at least a multiplicative factor of $\log^{-0.4} n$.
	\end{itemize}
\end{lemma}

\begin{proof}
	Fix any $r\in [\maxr]$, and consider any particular period of $2(r+1)$ rounds, where node $v$ sends $\Deal$\ values to clusters in $NC'_r(v)$ at the start of the period, and receives \Budget\ values and updates its weight at the end of the period. Let $w(v)$ and $d'_r(v)$ denote the values at the start of the period, and $\hat w(v)$, $\hat d'_r(v)$ denote the updated values at the end of the period.
	
	Assume that $w(v)$ is not reduced by at least $\frac{ w_i(v) }{\log^{0.7} n}$, i.e. $\hat w(v)> w(v) - \frac{ w_i(v)}{\log^{0.7} n}$. Then, there are at most $d_r(v)\log^{-0.4} n$ clusters $C\in NC'_r(v)$ for which $\Budget(v,C)=\Deal(v,C)$. In all other clusters $C$, for all $(v,w)\in E_C$, $w$ is left with $0$ weight and therefore joins the vertex cover and terminates, meaning that $C$ ceases to contain non-terminated neighbors of $v$ and leaves $NC'_r(v)$. So, $\hat d'_r(v) \le d'_r(v)\log^{-0.4} n$.
\end{proof}

Now we can show that all nodes that began phase 2 as active will terminate during the phase:

\begin{lemma}\label{lem:progressP2b}
	After $O(\frac{\log n}{\log^2\log n})$ rounds of phase 2, no nodes remain active.
\end{lemma}

\begin{proof}
	For any fixed node $v$ and $r\in [\maxr]$, $d'_r(v)$ is initially at most $e^{\frac{90\log n}{(r+1)\log\log n}}$, by the properties of the clustering. In each of the $\Theta(\frac{\log n}{(r+1)\log^2\log n})$ periods of $2(r+1)$ rounds within the phase, either $w(v)$ is reduced by at least $\frac{ w_i(v)}{\log^{0.7} n}$, or $d'_r(v)$ reduces by a multiplicative factor of at least $\log^{-0.4} n$. Therefore, choosing a sufficiently large constant in the $O(\frac{\log n}{\log^2\log n})$ running time of the phase, either $v$ becomes inactive during the phase, $w(v)$ falls below $w_i(v)-\frac{ w_i(v) }{\log^{0.7} n}\cdot \frac{\log n}{(r+1)\log^2\log n}< 0$ (which is not possible), or $d'_r(v) \le e^{\frac{90\log n}{(r+1)\log\log n}} \cdot \left(\log^{0.5}n\right)^{- \frac{1000\log n}{(r+1)\log^2\log n}}<1$, i.e. $d'_r(v)=0$. Since this holds for all $r\in [\maxr]$, if $v$ remains active it must have no remaining non-terminated neighbors. Then, $v$ will output $\NotInCover$ and terminate.
\end{proof}

\paragraph{Overall Analysis}

It remains to prove the main result, Theorem \ref{thm:MWVC}.

\begin{proof}[Proof of Theorem \ref{thm:MWVC}]
	By Lemmas \ref{lem:progressP1b} and \ref{lem:progressP2b}, after both phases of a stage, all nodes either terminate or become inactive; in the latter case their weight decreases by a multiplicative factor of at least $\log^{-0.1}n$. So, after $1+O(\log_{\log^{-0.1}n} \eps') = 1+ O(\frac{\log 1/\eps}{\log\log n})$ stages, the weight of any non-terminated node $v$ must decrease below $\eps' w_0(v)$, and $v$ will therefore join the verex cover and terminate. So, all nodes terminate within $\left(1+O(\frac{\log 1/\eps}{\log\log n})\right) \cdot \frac{\log n}{\log^2\log n} = O(\frac{\log n}{\log^2\log n} + \frac{\log n\log 1/\eps}{\log^3\log n})$ rounds.
	
	The correctness of the minimum weight vertex cover follows from the local ratio template (Theorem \ref{thm:localratio}): we iteratively (implicitly) reduce weights across edges, correspondingly reducing node weights, while ensuring all node weights remain non-negative (Lemma \ref{lem:posweights}). Nodes only join the vertex cover when their weight is at most $\eps' w_0(v)$, and nodes only output $\NotInCover$ when all of their neighbors are in the cover.
\end{proof}

The round complexity can be simplified for $\eps = \log^{-O(1)} n$:

\begin{corollary}
	For any $\eps = \log^{-O(1)} n$, a $2+\eps$-approximate minimum weight vertex cover can be computed in $O(\frac{\log n}{\log^2\log n})$ rounds of \local, succeeding with high probability in $n$.
\end{corollary}

\section{Graph Decomposition}

In this section we show the graph decomposition result that forms the backbone of our algorithms. Our graph decomposition is a generalization of a well-known randomized procedure of Miller, Peng, and Xu \cite{MPX13} (see Algorithm \ref{alg:part}): nodes each independently generate `shift' variables, and then join the cluster of the node with the lowest `shifted' distance value. This procedure can be show to produce low-diameter clusters with few inter-cluster edges, and it is straightforward to implement in \local\ and \congest\ (with round complexity equal to the largest possible cluster radius). The clustering produced by the procedure is highly related to other forms of network decomposition, such as the $(C,D)$-network decompositions that were introduced by Awerbuch, Luby, Goldberg and Plotkin \cite{ALGP89} and are now one of the most fundamental tools in distibuted graph algorithms (see e.g. \cite{RG20}). Indeed, Algorithm \ref{alg:part} can be adapted to give an $(O(\log n),O(\log n))$-network decomposition in $O(\log n)$ rounds of \local\  with high probability\footnote{Run $O(\log n)$ instances of Algorithm \ref{alg:part} in parallel, using an exponential distribution with constant parameter, and have each node $v$'s color be the index of the instance maximizing the gap between the `winning' $\delta_u - dist(u,v)$ value and the next highest. One can show that with high probability these gaps will all be at at least $2$, in which case the result is a valid $(O(\log n),O(\log n))$-network decomposition.}. However, we are not aware of this, or related notions of clustering/decomposition, being previously used to give \emph{sublogarithmic} distributed algorithms.

\RestyleAlgo{boxruled}

\begin{algorithm}
	\caption{\textsc{Partition$(F)$}\cite{MPX13}}
	\label{alg:part}
	\begin{itemize}
		\item Each node $v$ independently samples $\delta_v \sim F$.
		\item Nodes $v$ join the cluster of the node maximizing $\delta_u - dist(u,v)$.
	\end{itemize}
\end{algorithm}

Miller, Peng, and Xu used as $F$ an exponential distribution, i.e. with cumulative distribution function $F(x) = 1-e^{\lambda x}$ for some parameter $\lambda$. One major reason for this choice is that the exponential distribution has a property called `memorylessness' ($\Prob{X>x+c\vert X>c}= \Prob{X>x}$) that allows a clean analysis of how the clustering behaves. The many subsequent applications of the procedure (including on graph spanners \cite{MPVX15,EN18}, expander decomposition and triangle enumeration \cite{CPSZ21}, minimum spanning tree \cite{A+25}, and radio network communication \cite{HW16,CD21, C+18,D23}) have, to our knowledge, all used the exponential distribution (or highly related distributions, such as the geometric) and have employed the memorylessness property in their analysis. 

Here, we show how to analyze Algorithm \ref{alg:part} under an arbitrary distribution $F$, i.e. without requiring the memorylessness property. We focus on the case where clusters have low (sublogarithmic) diameter, and the number of clusters adjacent to each node is relatively high, but less than polynomial in $n$.

In this work, our shift variable distributions will be of the following form:

\begin{definition}\label{def:gcdf}[Shift variable distributions]
	Let $F: \mathbb R_{\ge 0}\rightarrow [0,1]$ be a cumulative distribution function, with the following properties:
	
	\begin{itemize}
		\item $F$ is non-decreasing and absolutely continuous.
		\item $F(0) = 0$.
		\item $\lim_{x\rightarrow \infty} F(x) = 1$.
	\end{itemize}
	
	We define also $G: \mathbb R_{\ge 0}\rightarrow [0,1]$ as $G(x) = 1-F(x)$; when defining specific functions to use, it is often more convenient to do so by defining $G$. Let the probability density function $f:\mathbb R_{\ge 0}\rightarrow \mathbb R_{\ge 0}$ be the derivative of $F$.
\end{definition}

When applying Algorithm \ref{alg:part} with some shift variable distribution $F$, for some fixed node $v$, let $d_w$ denote $dist(w,v)$. Then, define  \[C_{q}^{\diff}:= |\{w\ne v :\delta_w\ge \max\{\delta_{CC(v)}-d_{CC(v)}-\diff + d_w,q\}   \}|\enspace.\]

That is, $C_{q}^{\diff}$ is the number of nodes $w$ for which $\delta_w - d_w$ is within $\diff$ of the `winning' value of $v$'s cluster center, and $\delta_w$ is at least $q$. While we allow any $\diff\in \mathbb N$ for generality, in our applications we are particularly concerned with $C_{q}^{2}$. Notice that $C_{q}^{2}$ is an upper bound on the number of clusters of radius at least $q$ adjacent to $v$ (other than  $v$'s own cluster). (Here `radius' means the maximum distance from that cluster's center to any of its members.) This is because, since $F$ has domain $\mathbb R_{\ge 0}$, the radius of any cluster $C$ is at most the shift variable generated by $C$'s cluster center (since otherwise the furthest member from the center would have been its own cluster in preference to joining $C$). Also, for a cluster centered at $w$ to be adjacent to $v$ (but not contain it), we must have $\delta_{CC(v)}-d_{CC(v)}-2 \le \delta_{w}-d_{w} < \delta_{CC(v)}-d_{CC(v)}$ (where $CC(v)$ denotes the cluster center of $v$'s cluster).

The remainder of this section will be dedicated to proving bounds on $C_{q,\diff}$ based on $F$. Since we will want a high-probability bound on this quantity, rather than bounding $\Exp{C_{q,\diff}}$, we will instead bound $\Exp{e^{\gamma(q) C_{q,\diff}}}$ for some $\gamma(q)$.

\begin{theorem}\label{thm:moment}
	Consider running Algorithm \ref{alg:part} with some shift variable distribution $F$ satisfying Definition \ref{def:gcdf}. Let $\mult^{\diff}(q):=\sup_{x\ge q} \left(\frac{G(x)}{G(x+2)}-1\right)$, and let $\gamma(q)$ be such that $e^\gamma(q) - 1\le \frac{F( 1)}{2\mult^{\diff}(q)}$. Then, for any node $v$, $\Exp{e^{\gamma(q) C_{q,\diff}}}\le 2e$.
\end{theorem}

\begin{proof}
	
	We begin with the following expression for $\Exp{e^{\gamma(q) C_{q,\diff}}}$, using the law of total probabilities over the possible cluster centers for $v$:
	
	\begin{align*}
	\Exp{e^{\gamma(q) C_{q,\diff}}} &= \sum_{u\in V} \Prob{u \text{ is $v$'s cluster center}} \cdot \Exp{e^{\gamma(q) C_{q,\diff}} \vert u \text{ is $v$'s cluster center}}\\
	&= \int_{0}^{\infty}\sum_{u\in V} \Prob{u \text{ is $v$'s cluster center}\vert \delta_u-d_u = p}\\ 
	&\hspace{1in}\cdot \Exp{e^{\gamma(q) C_{q,\diff}} \vert u \text{ is $v$'s cluster center} \land  \delta_u-d_u = p} f(p+d_u) dp\enspace.
	\end{align*}
	
	We bound the two main terms separately. First:
	
	\begin{align*}
	\Prob{u \text{ is $v$'s cluster center}\vert \delta_u -d_u = p} &= \prod_{w\in V\setminus \{u\}}\Prob{\delta_w-d_w < p}\\
	&\le \frac{1}{F(p+d_u)}\prod_{w\in V}F(p+d_w)\\
	& =   \frac{1}{F(p+d_u)}\prod_{w\in V}\left(1-G(p+d_w)\right)\enspace.\\
	\end{align*}
	
	Using the inequality $1-x\le e^{-x}$ for all $x\ge 0$, and separating out the term incurred by $v$ itself, with $d_v = 0$,
	
	\begin{align*}
	\prod_{w\in V}\left(1-G(p+d_w)\right)&\le \frac{1-G(p)}{e^{-G(p)}} \cdot\prod_{w\in V}e^{-G(p+d_w)}\\
	&=F(p){e^{G(p)}} \cdot e^{-\sum_{w\in V}G(p+d_w)}\\
	&\le e F(p) \cdot e^{-\sum_{w\in V}G(p+d_w)}\enspace.
	\end{align*}
	
	Using this bound,
	\begin{align*}
	\Prob{u \text{ is $v$'s cluster center}\vert \delta_u -d_u = p} &\le \frac{e^2F(p) }{F(p+d_u)}  \cdot e^{-\sum_{w\in V}G(p+d_w)}\\
	&\le e^{1-\sum_{w\in V}G(p+d_w)}\enspace.
	\end{align*}

	Bounding the $\Exp{e^{\gamma(q) C_{q,\diff}} \vert u \text{ is $v$'s cluster center} \land  \delta_u-d_u = p}$ term, we note that $v$ cannot be adjacent to a cluster centered at itself, since if there is such a cluster, $v$ will be in it (and $C_{q,d}$ does not count $v$'s own cluster). That is why, in the following bound, the sum is over $V\setminus\{u,v\}$:
	
	\begin{align*}
	&\Exp{e^{\gamma(q) C_{q,d}} \vert u \text{ is $v$'s cluster center} \land  \delta_u -d_u= p}\\
	&\hspace{1in}= \Exp{e^{\gamma(q) \sum_{w\in V\setminus \{u,v\}}\mathbf 1_{\{ \delta_w\ge \max\{\delta_{u}-d_{u}-\diff + d_w,q\} \vert u \text{ is $v$'s cluster center} \land  \delta_u-d_u = p\}}}}\\
	&\hspace{1in}\le \Exp{e^{\gamma(q) \sum_{w\in V\setminus \{u,v\}}\mathbf 1_{\{\delta_w\ge \max\{p-\diff + d_w,q\}\vert \delta_w \le p  + d_w\}}}}\\ 
	&\hspace{1in}=\prod_{w\in V\setminus \{u,v\}}\Exp{e^{\gamma(q) \mathbf 1_{\{\delta_w\ge \max\{p-\diff + d_w,q\}vert \delta_w \le p  + d_w\}}}},
	\end{align*}
	
	where the last equality holds since the indicator variables for each $w$ are independent. For a Bernoulli variable $X$ that occurs with probability $\mathbf p$, we use the following standard bound (often used in proof of the Chernoff bound):
	
	\[\Exp{e^{\gamma(q) X}} = (1-\mathbf p)e^0 + \mathbf p e^{\gamma(q)} = 1+\mathbf p(e^{\gamma(q)}-1) \le e^{\mathbf p(e^{\gamma(q)} - 1)}\enspace.\]
	
	Plugging this in for $X=\mathbf 1_{\{\delta_w\ge \max\{p-\diff + d_w,q\}\vert \delta_w \le p  + d_w\}}$,
	
	\begin{align*}
	&\Exp{e^{\gamma(q) C_{q,\diff}} \vert u \text{ is $v$'s cluster center} \land  \delta_u -d_u= p}\\
	&\hspace{1in}=\prod_{w\in V\setminus \{u,v\}} e^{\Prob{\delta_w\ge \max\{p-\diff + d_w,q\}\vert \delta_w \le p  + d_w}(e^{\gamma(q)} - 1)}\\
	&\hspace{1in}\le e^{(e^{\gamma(q)} - 1)\sum_{w\in V\setminus \{v\}}\Prob{\delta_w\ge \max\{p-\diff + d_w,q\}\vert \delta_w \le p  + d_w}}\enspace.
	\end{align*}
	
	For any $w\ne v$, we can bound 
	
	\begin{align*}
	\Prob{\delta_w\ge \max\{p-\diff + d_w,q\}\vert \delta_w \le p  + d_w} &\le  
	\frac{F(p  + d_w)-F(\max\{q, p -\diff+d_w\})}{F(p  + d_w)}\\
	&\le \frac{1}{F(1)}\cdot \left(G(\max\{q, p -\diff+d_w\})-G(p  + d_w)\right)\\
	&= \frac{G(p  + d_w)}{F(1)}\cdot \left(\frac{G(\max\{q, p -\diff+d_w\})}{G(p  + d_w)}-1\right)\\
	&\le\frac{G(p  + d_w)}{F(1)}\cdot\sup_{x\ge q} \left(\frac{G(x)}{G(x+2)}-1\right)\\
	&\le \frac{\mult^{\diff}(q) G(p  + d_w)}{F(1)}\enspace.
	\end{align*}
	
	Then,
	
	\begin{align*}
	\Exp{e^{\gamma(q) C_{q,\diff}} \vert u \text{ is $v$'s cluster center} \land  \delta_u -d_u= p}
	&\le e^{(e^{\gamma(q)} - 1)\sum_{w\in V\setminus \{v\}}\Prob{\delta_w\ge \max\{p-\diff + d_w,q\}\vert \delta_w \le p  + d_w}}\\
	&\le e^{\frac{\mult^{\diff}(q) (e^{\gamma(q)} - 1)}{F(1)}\sum_{w\in V\setminus \{v\}}G(p  + d_w)}\\
	&\le e^{\frac 12\sum_{w\in V}G(p  + d_w)}\enspace.
	\end{align*}
	
	Putting the bounds together,
	
	\begin{align*}
	\Exp{e^{\gamma(q) C_{q,\diff}}} &= \int_{0}^{\infty}\sum_{u\in V} \Prob{u \text{ is $v$'s cluster center}\vert \delta_u-d_u = p}\\ 
	&\hspace{1in}\cdot \Exp{e^{\gamma(q) C_{q,\diff}} \vert u \text{ is $v$'s cluster center} \land  \delta_u-d_u = p} f(p+d_u) dp\\
	&\le  \int_{0}^{\infty}\sum_{u\in V} e^{1-\sum_{w\in V}G(p+d_w)} \cdot  e^{\frac 12\sum_{w\in V}G(p  + d_w)}f(p+d_u) dp\\ 
	&=  e\int_{0}^{\infty}e^{-\frac12 \sum_{w\in V}G(p+d_w)}\cdot \sum_{u\in V}  f(p+d_u) dp
	\enspace.
	\end{align*}
	
	We now evaluate the integral by substitution, setting our new function $H(p) := -\sum_{w\in V}G(p +d_w)$. Notice that the derivative $h$ of $H$ is $\sum_{w\in V}f(p +d_w)$. So, we can rephrase the above inequality as:
	
	\begin{align*}
	\Exp{e^{\gamma(q) C_{q,\diff}}} 
	&\le e   \int_{0}^{\infty} e^{\frac12 H(p)}\cdot  h(p) dp\enspace.
	\end{align*}
	
	Integration by substitution then gives:
	
	\begin{align*}
	\Exp{e^{\gamma(q) C_{q,\diff}}} 
	&\le e   \int_{H(0)}^{H(\infty)} e^{\frac z2} dz\\
	&= e \left[2 e^{\frac z2}\right]_{H(0)}^{0}\\
	&\le 2e\enspace.
	\end{align*}
\end{proof}

For our applications, we will use $\textsc{Partition}(F)$ with the following shift variable distribution $F$:

\[F(x) =  1-(1+x)^{-\alpha}; \]

where $\alpha = \frac{\const\log n}{\log\log n}$, for some constant (to be fixed later) $\const \ge 4$. Applying Algorithm \ref{alg:part} with this distribution allows us to prove Lemma \ref{lem:localpart}:

\begin{proof}[Proof of Lemma \ref{lem:localpart}]
	The radius of any cluster is at most the shift variable of its cluster center. The probability that any particular shift variable $\delta_v$ exceeds $\log^{3/\const} n$ is at most $G(\log^{3/\const} n) \le  \log^{\frac{-3\alpha}{\const}} n = n^{-3}$. Therefore, by a union bound over all nodes $v$, with high probability all shift variables are at most $\log^{3/\const} n$ and all clusters are of radius at most $\log^{3/\const} n$. The algorithm therefore only needs to run for $\log^{3/\const} n$ rounds before all nodes know the identity of their cluster center with high probability. This proves the first part of the lemma statement.
	
	For the second part, we fix any node $v$ and bound $C_{q,2}$ which as discussed is an upper bound on the number of clusters of radius at least $q$ adjacent to $v$, excluding $v$'s own cluster. $C_{q,2}+1$ is therefore an upper bound on the quantity including $v$'s cluster. Note that we now need only consider $q\le \log^{3/\const} n$ (since for higher $q$, with high probability $C_{q,2}$ = 0).
	
	We will apply Theorem \ref{thm:moment}, but we must first bound $\mult^2(q)$ and choose a suitable $\gamma(q)$. We have 
	\begin{align*}
	\mult^2(q) &= \sup_{x\ge q}\frac{G(x)}{G(x+2)} -1\\
	&=  \sup_{x\ge q}\left(\frac{1+x}{3+x}\right)^{-\alpha}\\
	&=  \sup_{x\ge q}\left(1+\frac{2}{1+x}\right)^{\alpha}\\
	&\le  \sup_{x\ge q}e^{\frac{2\alpha}{1+x}}\\
	&= e^{\frac{2\alpha}{1+q}}.
	\end{align*}
	
	Since $q\le  \log^{3/\const} n <\frac{\alpha}{9}$, $e^{\frac{2\alpha}{1+q}}\ge e^{17}$.
	
	To apply Theorem \ref{thm:moment}, we require $\gamma(q)$ such that $e^{\gamma(q)} - 1\le \frac{F( 1)}{2\mult^{2}(q)}$. We choose $\gamma(q) = \frac{1}{6}e^{\frac{-2\alpha}{1+q}}$. Then, since $\gamma(q) <1$,
	
	\begin{align*}
	e^{\gamma(q)} -1 \le 2\gamma(q)
	= \frac 13 e^{\frac{-2\alpha}{1+q}}
	\le \frac{1-2^{-\alpha}}{2e^{\frac{2\alpha}{1+q}}-1}
	\le\frac{F( 1)}{2\mult^{2}(q)}\enspace.
	\end{align*}
	
	So, applying Theorem \ref{thm:moment}, $\Exp{e^{\gamma(q) C_{q,2}}} \le 2e$.
	
	By Markov's inequality, 
	\begin{align*}
	\Prob{e^{\gamma(q) C_{q,2}} \ge 5n^3} \le \frac{2e}{5n^3}\le n^{-3}\enspace.
	\end{align*}
	
	So, with probability at least $1-n^{-3}$, we have $e^{\gamma(q) C_{q,2}} <5n^3$, and therefore
	
	\begin{align*}
	C_{q,2} &\le \frac{\ln (5n^3)}{\gamma(q)}\\
	&< 6\ln (5n^3)e^{\frac{2\alpha}{q+1}}\\
	&< e^{\frac{3\alpha}{q+1}} -1&&\text{(since $q\le \log^{3/\const} n$).}
	\end{align*}
	
	By a union bound, this occurs for all nodes $v$ with probability at least $1-n^{-2}$. Then, all nodes have at most $e^{\frac{3\alpha}{q+1}} $ adjacent clusters of radius at least $q$, proving the second part of the lemma statement.
\end{proof}

\section{Conclusions and Open Problems}
In this work, we showed a new graph decomposition result that allowed us to improve the complexities of approximate maximum matching and minimum weighted vertex cover as a function of $n$. This opens up new possibilities for improved complexities on high degree graphs. Our $\log\log n$-factor improvement demonstrates that the $O(\frac{\log\Delta}{\log\log \Delta})$ upper bounds for these problems are not tight for high $\Delta$, but there is still a large gap between upper and lower bounds; one open question is how close one can get to the $\Omega(\sqrt{\frac{\log n}{\log\log n}})$ lower bound of Kuhn, Moscibroda and Wattenhofer \cite{KMW16}.

Another major question is whether our techniques here can be extended to more difficult problems. For example, can one similarly use a decomposition to make $\log\log n$-factor improvements for:

\begin{itemize}
	\item Approximate maximum independent set, $1+\eps$-approximate maximum matching, and $2+\eps$-approximate \emph{weighted} maximum matching, following \cite{BCGS17}?
	\item Maximal matching and maximal independent set? 
\end{itemize}

Any $o(\log n)$-round algorithm for the latter problems would be a very interesting result, since they are often considered the most central problems in the \local\ model, and no improvement on their algorithmic complexities as a function of $n$ has been made since the $O(\log n)$-round algorithms of Israeli and Itai, and Luby, over 40 years ago \cite{II86, L86}. 

\section{Acknowledgements}
	Peter Davies-Peck is supported by EPSRC New Investigator Award UKRI155.

\newcommand{\Proc}{Proceedings of the\xspace}
\newcommand{\STOC}{Annual ACM Symposium on Theory of Computing (STOC)}
\newcommand{\FOCS}{IEEE Symposium on Foundations of Computer Science (FOCS)}
\newcommand{\SODA}{Annual ACM-SIAM Symposium on Discrete Algorithms (SODA)}
\newcommand{\AISTATS}{International Conference on Artificial Intelligence and Statistics}
\newcommand{\COCOON}{Annual International Computing Combinatorics Conference (COCOON)}
\newcommand{\DISC}{International Symposium on Distributed Computing (DISC)}
\newcommand{\ESA}{Annual European Symposium on Algorithms (ESA)}
\newcommand{\ICALP}{Annual International Colloquium on Automata, Languages and Programming (ICALP)}
\newcommand{\ICML}{International Conference on Machine Learning}
\newcommand{\ICLR}{International Conference on Learning Representations}

\newcommand{\IPL}{Information Processing Letters}
\newcommand{\JACM}{Journal of the ACM}
\newcommand{\JALGORITHMS}{Journal of Algorithms}
\newcommand{\JCSS}{Journal of Computer and System Sciences}
\newcommand{\NEURIPS}{Conference on Neural Information Processing Systems}
\newcommand{\PODC}{Annual ACM Symposium on Principles of Distributed Computing (PODC)}
\newcommand{\SICOMP}{SIAM Journal on Computing}
\newcommand{\SPAA}{Annual ACM Symposium on Parallelism in Algorithms and Architectures (SPAA)}
\newcommand{\STACS}{Annual Symposium on Theoretical Aspects of Computer Science (STACS)}
\newcommand{\TALG}{ACM Transactions on Algorithms}
\newcommand{\TCS}{Theoretical Computer Science}
\bibliographystyle{plain}

\bibliography{graphdecomp}
\hide{
	\section{General version - basic}\label{sec:basic}

	\begin{theorem}\label{thm:part1}
		Let $F$ be a cumulative distribution function as in \Cref{def:cdf}. Then, Algorithm \ref{alg:part} produces a clustering with the following properties, with high probability:
		
		\begin{itemize}
			\item For any node $v$, and for $q\in \mathbb N$, let $C_{q,d}$ denote the number of cluster of radius at least $q$ that are within distance $d$ of $v$ (not including the cluster that $v$ itself is in). Then, $\Exp{C_{q,d}}$ is at most $e^2 (\ratiod(q)-1)$.
			\item The maximum radius of any cluster is at most $G^{-1}(n^{-3})$, and therefore maximum strong diameter is at most $2G^{-1}(n^{-3})$.
			\item Algorithm \ref{alg:part} can be run in $O(G^{-1}(n^{-3}))$ rounds of \local, \congest, or \bcongest and in $O(G^{-1}(n^{-3}))\log^2 n$ rounds in radio networks.
		\end{itemize}
	\end{theorem}
	
	\begin{proof}
		For the first point, we analyze the behavior around a fixed node $v$. For any node $u$, let $d_u$ denote $dist(u,v)$. Notice that, denoting $v$'s cluster center as $u$, the center $w$ of any cluster within distance $d$ of $v$ must satisfy $\delta_w-d_w \ge \delta_u - d_u - 2d$. Otherwise, every node $z$ within distance $d$ of $v$ would have $\delta_w-dist(z,w) < \delta_u - dist(z,u)$ and so would not join $w$'s cluster.
		
		We begin with the following expression for $\Exp{C_{q,d}}$, using the law of total probabilities over the possible cluster centers for $v$:
		
		\begin{align*}
		\Exp{C_{q,d}} &= \sum_{u\in V} \Prob{u \text{ is $v$'s cluster center}} \cdot \Exp{C_{q,d} \vert u \text{ is $v$'s cluster center}}\\
		&= \int_{0}^{\infty}\sum_{u\in V} \Prob{u \text{ is $v$'s cluster center}\vert \delta_u-d_u = p}\\ 
		&\hspace{1in}\cdot \Exp{C_{q,d} \vert u \text{ is $v$'s cluster center} \land  \delta_u-d_u = p} f(p+d_u) dp\enspace.
		\end{align*}
		
		We bound the two main terms separately. First:
		
		\begin{align*}
		\Prob{u \text{ is $v$'s cluster center}\vert \delta_u -d_u = p} &= \prod_{w\in V\setminus \{u\}}\Prob{\delta_w-d_w < p}\\
		&\le \frac{1}{F(p+d_u)}\prod_{w\in V}F(p+d_w)\\
		& =   \frac{1}{F(p+d_u)}\prod_{w\in V}\left(1-G(p+d_w)\right)\enspace.\\
		\end{align*}
		
		$V$ contains $v$ itself, with $d_v = 0$, and we may assume it contains at least one node $w$ with $d_w=1$ (since otherwise $v$ is disconnected and trivially is not within distance $d$ of any cluster other than its own). So, using the inequality $1-x\le e^{-x}$ for all $x\ge 0$,
		
		\begin{align*}
		\prod_{w\in V}\left(1-G(p+d_w)\right)&\le \frac{1-G(p)}{e^{-G(p)}} \cdot \frac{1-G(p+1)}{e^{-G(p+1)}}\cdot\prod_{w\in V}e^{-G(p+d_w)}\\
		&=F(p){e^{G(p)}} \cdot F(p+1){e^{G(p+1)}}\cdot e^{-\sum_{w\in V}G(p+d_w)}\\
		&\le e^2 F(p) \cdot F(p+1)\cdot e^{-\sum_{w\in V}G(p+d_w)}\enspace.
		\end{align*}
		
		Using this bound,
		\begin{align*}
		\Prob{u \text{ is $v$'s cluster center}\vert \delta_u -d_u = p} &\le \frac{e^2F(p) \cdot F(p+1)}{F(p+d_u)}  \cdot e^{-\sum_{w\in V}G(p+d_w)}\\
		&\le e^2  F(p+1) \cdot e^{-\sum_{w\in V}G(p+d_w)}\enspace.
		\end{align*}

		Bounding the $\Exp{C_{q,d} \vert u \text{ is $v$'s cluster center} \land  \delta_u-d_u = p}$ term, we note that $v$ cannot be adjacent to a cluster centered at itself, since if there is such a cluster, $v$ will be in it (and $C_{q,d}$ does not count $v$'s own cluster). That is why, in the following bound, the sum is over $V\setminus\{u,v\}$:
		\begin{align*}
		&\Exp{C_{q,d} \vert u \text{ is $v$'s cluster center} \land  \delta_u -d_u= p}\\
		&\hspace{1in}\le \sum_{w\in V\setminus \{u,v\}}\Prob{\delta_w\ge q \land \delta_w-d_w \ge \delta_u-d_u  -2d\vert u \text{ is $v$'s cluster center} \land  \delta_u-d_u = p}\\
		&\hspace{1in}= \sum_{w\in V\setminus \{v\}}\Prob{\delta_w\ge \max\{q, p -2d+d_w\}\vert \delta_w \le p  + d_w}\\
		&\hspace{1in}= \sum_{w\in V\setminus \{v\}}\frac{F(p  + d_w)-F(\max\{q, p -2d+d_w\})}{F(p  + d_w)}\\
		&\hspace{1in}\le  \sum_{w\in V\setminus \{v\}}\frac{F(p  + d_w)-F(\max\{q, p -2d+d_w\})}{F(p  + 1)}\\
		&\hspace{1in}\le \frac{1}{F(p  + 1)} \sum_{w\in V}F(p  + d_w)-F(\max\{q, p -2d+d_w\})\\
		&\hspace{1in}= \frac{1}{F(p  + 1)} \sum_{w\in V}G(\max\{q, p -2d+d_w\})-G(p  + d_w)\\
		&\hspace{1in}=\frac{1}{F(p  + 1)} \sum_{w\in V} G(p  + d_w)\cdot \left(\frac{G(\max\{q, p -2d+d_w\})}{ G(p  + d_w)}-1\right)\\
		&\hspace{1in}\le  \frac{ \ratiod(q)-1}{F(p  + 1)} \sum_{w\in V} G(p  + d_w)\enspace.
		\end{align*}

		So, combining these two bounds, we now have:
		
		\begin{align*}
		\Exp{C_{q,d}} &= \int_{0}^{\infty}\sum_{u\in V} \Prob{u \text{ is $v$'s cluster center}\vert \delta_u = p} \cdot \Exp{C_q \vert u \text{ is $v$'s cluster center} \land  \delta_u = p} f(p) dp\\
		&\le e^2  (\ratiod(q)-1) \int_{0}^{\infty}\sum_{u\in V} f(p+d_u) \cdot e^{-\sum_{w\in V}G(p +d_w)}\cdot \sum_{w\in V}G(p  + d_w)  dp\enspace.
		\end{align*}
		
		We now evaluate the integral by substitution, setting our new function $H(p) := -\sum_{w\in V}G(p +d_w)$. Notice that the derivative $h$ of $H$ is $\sum_{w\in V}f(p +d_w)$. So, we can rephrase the above inequality as:
		
		\begin{align*}
		\Exp{C_{q,d}}
		&\le -e^2  (\ratiod(q)-1) \int_{0}^{\infty}h(p) \cdot e^{H(p)}\cdot H(p) dp\enspace.
		\end{align*}
		
		Integration by substitution then gives:
		
		\begin{align*}
		\Exp{C_{q,d}}
		&\le -e^2  (\ratiod(q)-1) \int_{H(0)}^{H(\infty)} u e^{u} du\\
		&= -e^2  (\ratiod(q)-1) \left[(u-1) e^{u}\right]_{H(0)}^{0}\\
		&\le -e^2(\ratiod(q)-1)(0-1) e^{0}\\
		&=e^2  (\ratiod(q)-1)\enspace.
		\end{align*}
		
		For the second point, notice that with probability at least $1-n^{-2}$, no node $v$ chooses $\delta_v \ge G^{-1}(n^{-3})$, by a union bound. So, with probability at least $1-n^{-2}$, all clusters have radius at most $G^{-1}(n^{-3})$, and so strong diameter at most $2G^{-1}(n^{-3})$.
		
		For the third point, clearly by the above the algorithm runs in $O(G^{-1}(n^{-3}))$ rounds of \local. To implementent the algorithm in \congest and \bcongest, each node $v$ sends the current highest value $\delta_w-dist(w,v)$ it knows to all neighbors each round, along with the corresponding ID $w$ - after $O(G^{-1}(n^{-3}))$, this value will be that of its final cluster center with high probability. For radio networks, see the radio network implementation of CITATIONS (\cite{MPX13} by \cite{HW16}).
	\end{proof}
}

\end{document}